\documentclass[conference]{IEEEtran}
\IEEEoverridecommandlockouts
\usepackage{cite}  
\usepackage{fancyhdr}
\usepackage[utf8]{inputenc}
\usepackage[T1]{fontenc}

\makeatletter

\newcommand{\Rmnum}[1]{\expandafter\@slowromancap\romannumeral #1@}
\makeatother

\usepackage{amssymb}
\usepackage{mathtools}
\usepackage{amsfonts}
\usepackage{amsthm}
\usepackage{amsmath}
\usepackage{mathrsfs}
\usepackage{graphicx}
\usepackage{bm}
\usepackage{empheq}

\usepackage{url}
\usepackage{multirow}
\usepackage{cases}
\usepackage{nomencl}
\makenomenclature
\usepackage{color}
\usepackage{xcolor}

\newtheorem{assumption}{Assumption}

\newtheorem{remark}{Remark}
\newtheorem{theorem}{Theorem}
\newtheorem{lemma}{Lemma}

\newtheorem{definition}{Definition}

\title{\LARGE \bf
	Distributionally Robust Equilibria over the Wasserstein Distance for Generalized Nash Game 
}

\author{Yixun~Wen, Yulong Gao and Boli~Chen
    \thanks{Y. Wen and B. Chen are with the Department of Electronic and Electrical Engineering, University College London, UK, WC1E 6BT {\tt\small (yixun.wen.22@ucl.ac.uk; boli.chen@ucl.ac.uk)}.}%
    \thanks{Yulong Gao is with the Department of Electrical and Electronic Engineering, Imperial College London, London, UK SW7 2AZ {\tt\small(
yulong.gao@imperial.ac.uk)}.}
}

\fancypagestyle{myheader}{
    \fancyhf{} % 清空默认的页眉页脚内容
    \fancyhead[L]{ } % 左侧页眉内容
    \fancyhead[C]{\footnotesize This work has been submitted to IEEE CDC 2025. Copyright may be transferred without notice, after which this version might no longer be accessible.} % 中央页眉内容
    \fancyhead[R]{ } % 右侧页眉内容
    \setlength{\headsep}{25pt}    % 设置页眉与正文之间的距离
}

\begin{document}
% LTeX: enabled=true

\pagestyle{myheader}
\begin{titlepage}
    \centering
    \vspace*{\fill}
    {\Huge \bfseries Copyright Statement \\[0.5cm]}

    {\large This work has been submitted to the CDC 2025. Copyright may be transferred without notice, after which this version might no longer be accessible. \\[2cm]}
    \vspace*{\fill}
\end{titlepage}

\clearpage

\maketitle

\begin{abstract}
Generalized Nash equilibrium problem (GNEP) is fundamental for practical applications where multiple self-interested agents work together to make optimal decisions. In this work, we study GNEP with shared distributionally robust chance constraints (DRCCs) for incorporating inevitable uncertainties. The DRCCs are defined over the Wasserstein ball, which can be explicitly characterized even with limited sample data. To determine the equilibrium of the GNEP, we propose an exact approach to transform the original computationally intractable problem into a deterministic formulation using the Nikaido-Isoda function. Specifically, we show that when all agents' objectives are quadratic in their respective variables, the equilibrium can be obtained by solving a typical mixed-integer nonlinear programming (MINLP) problem, where the integer and continuous variables are decoupled in both the objective function and the constraints. This structure significantly improves computational tractability, as demonstrated through a case study on the charging station pricing problem.
%We extend the reformulation of DRCC over the Wasserstein ball to a game-theoretic framework. To address the integer variables introduced during the reformulation process, we subsequently derive an equivalent convexified version of the problem. After obtaining the convexified version, we investigate the case where the objectives of all agents are quadratic and demonstrate that the equilibrium can be obtained by solving the corresponding mixed-integer linear programming (MINLP) problem.
\end{abstract}

\section{Introduction}
\label{sec:intro}
Numerous engineering applications require addressing multi-agent optimization problems and games among self-interested agents, including demand management in energy systems \cite{9353238}, electric vehicle (EV) charging coordination \cite{DEORI2018150}, and multi-robot control \cite{9992806}. These usually involve a generalized Nash equilibrium problem (GNEP) due to the shared constraints among agents, which generally reflect limitations in shared resources. 

The existence and uniqueness of equilibrium and solution methods have been studied for GNEP, especially under convexity assumptions. A comprehensive survey on the existence of Generalized Nash Equilibria (GNEs), with a focus on the jointly convex case, is provided in \cite{FaccGNE}.
% QVI theory
In \cite{HARKER199181, QVI-player, QVI-jointly}, the GNEP is analyzed using the quasi-variational inequality theory. 
% Patrick T. Harker was the first to analyze the GNEP using the quasi-variational inequality (QVI) theory \cite{HARKER199181}. He also demonstrated that, for a specific class of GNEPs, the QVI formulation can be reduced to a variational inequality (VI) problem.
% If a GNEP is player-convex, the GNE is equivalent to the solution of the corresponding QVI problem \cite{QVI-player}. Moreover, if the GNEP is jointly convex, the solution of the associated VI problem, known as the variational equilibrium, is guaranteed to be a GNE \cite{QVI-jointly}.
% NI func
The Nikaido-Isoda (NI) function is another important reformulation method, first introduced in \cite{NIfunc}, which provides existence results for GNEPs. This result was later extended to player-convex GNEPs in \cite{NI-func2}.
% KKT
Apart from the two reformulation methods discussed above, a Generalized Nash Equilibrium (GNE) can also be obtained by solving the Karush-Kuhn-Tucker (KKT) conditions, which consist of the concatenated KKT conditions of all players, provided certain conditions hold \cite{KKT1, FaccGNE}. 
% Specifically, if the GNEP is player-convex, any solution of the KKT conditions corresponds to a GNE \cite{KKT1}. Furthermore, in a jointly convex GNEP, the KKT conditions of the corresponding VI problem provide a solution to the KKT, where all agents share the same dual variables \cite{FaccGNE}.

% {\color{red}lack of robustness in the presence of uncertainty, then the role of DRCC, effective tools for dealing with stochastic uncertainties, add some citations, including our TCST paper..}

Most existing Generalized Nash Equilibrium Problem (GNEP) studies lack robustness considerations and may fall short when faced with uncertainties in real-world applications. Distributionally robust chance constraints (DRCCs), which serve as an effective tool for handling stochastic uncertainties, have been widely applied across various domains \cite{DRCC1,DRCC2,DRCC3}.
However, there has been limited research on analyzing GNE in the presence of DRCCs. The GNEP framework, where strategy sets are constrained by DRCCs defined over density-based and two-moment-based ambiguity sets, is explored in \cite{DRCCgame2}. Both of these formulations can be reformulated into convex deterministic problems.
Nevertheless, these two types of DRCCs perform poorly when the available sample size is limited. This is because density-based DRCCs fail to capture out-of-sample scenarios, while two-moment-based DRCCs, which rely on mean and variance, cannot fully characterize the underlying distribution when the sample size is insufficient.
Data-driven DRCCs defined over the Wasserstein ball \cite{dWRefo} can overcome the aforementioned shortcomings. This type of DRCC has been incorporated into the GNEP framework in \cite{DRCCgame} and subsequently relaxed into a convex form through the design of a penalty function. However, this reformulation is exact only if a tailored penalty function is used, which remains a significant challenge.

This study investigates the GNEP with joint DRCCs defined over Wasserstein balls, which better account for out-of-sample uncertainty and offer stronger reliability guarantees than density- or moment-based DRCCs, particularly in data-scarce settings [15]. The novelty of the paper is mainly twofold:  
1) We propose a novel and exact reformulation approach based on the Nikaido-Isoda function. Our approach is generally applicable, without finding a penalty function as in [16].
2) We demonstrate that when agents have quadratic objectives, the robust GNE can be solved by formulating it as a mixed-integer nonlinear programming (MINLP) problem, where the integer and continuous variables are decoupled in both the objective function and the constraints. This structure significantly improves computational tractability, as demonstrated through a case study on the charging station pricing problem.
% This study investigates the generalized Nash game with DRCCs defined over Wasserstein balls, which, in general, do not admit exact convex reformulations. To address the limitations of inexact convex approximations, we propose a novel approach to transform the original computationally intractable problem into a deterministic formulation using the Nikaido-Isoda function and a rigorous reformulation of the DRCC. Specifically, we show that when all agents' objectives are quadratic in their respective variables, the equilibrium can be obtained by solving a corresponding mixed-integer nonlinear programming (MINLP) problem. Furthermore, under certain convexity assumptions, this MINLP can be solved efficiently, as demonstrated by a case study on the charging station pricing problem.

The paper is organized as follows: Section \Rmnum{2} introduces preliminaries and gives the problem statement. Section \Rmnum{3} elaborates on the DRCC reformulation in the game-theoretical framework and establishes the existence condition for equilibrium under the quadratic assumption. Section \Rmnum{4} elaborates on the case study and simulation results, and Section \Rmnum{5} concludes this work.
%-----------------------------------------------
\section{Preliminaries and Problem Statement}
 Let $\mathbb{R}^n$, ${\mathbb{R}^n}_{\geq 0}$, ${\mathbb{R}^n}_{> 0}$ denote the set of real, the non-negative real, the strict positive real $n$-dimensional vectors. $\mathbb{Z}^n$ denote the set of integer $n$-dimensional vectors. $\mathbb{B}^n$ denote the set of binary $n$-dimensional vectors. $\|\cdot\|$ denotes a general norm and $||\cdot||_*$ denotes the dual norm. $\mathbf{e}_n$ and $\mathbf{0}_n$ indicate the column vectors with $n$ entries all equal to 1 and 0, respectively, i.e., $\mathbf{e}_n = [1,1,\cdots,1]^\top\in\mathbb{R}^n$ and $\mathbf{0}_n = [0,0,\cdots,0]^\top\in\mathbb{R}^n$. $\mathbb{I}_n\in\mathbb{R}^{n\times n}$ denotes the identity matrix. $[K]$ denotes the set $[K] = \{1,2,\cdots,K\}$. $(\cdot)^+$ denotes the operator that calculates the positive part. 
 The epigraph of $f(x): \mathbb{R}^n\rightarrow \mathbb{R}$ is defined as
    \begin{equation*}
        \text{epi}f  = 
        \left\{ (x,t)\in \mathbb{R}^{n+1}\left|
        x\in \text{dom}f, f(x)\leq t
        \right.
        \right\}.
    \end{equation*}
    The convex envelope of  $f: \mathbb{R}^n\rightarrow \mathbb{R}$ is defined as
    \begin{equation*}
        f^c(x) = \inf 
        \left\{ t\left|
        (x,t)\in \text{conv epi}f
        \right.
        \right\}.
    \end{equation*}
Finally, we use $\text{conv}S$ to denote the convex hull of $S$.
\begin{definition}[\cite{DisOpt}]\label{def:hole-free}
    A set $\mathcal{S}$ is called hole-free if 
    \begin{equation*}
        \mathcal{S} = \text{conv}\mathcal{S}\cap \mathbb{Z}^n,
    \end{equation*}
i.e., all integers in the convex hull of $\mathcal{S}$ are members of $\mathcal{S}$.
\end{definition}

%\subsection{Wasserstein Distance and Wasserstein Ball}
% For subsequent theoretical analysis in Section \ref{sec:DRCC}, 
%To facilitate the illustration of DRCC, we introduce the definitions of the Wasserstein distance.
\begin{definition}
The Wasserstein distance $d_w(\mathbb{P}_a,\mathbb{P}_b)$ between two distributions, equipped with a general norm $\|\cdot\|$, is defined as the minimal transportation cost of moving $\mathbb{P}_a$ to $\mathbb{P}_b$ under the premise that the cost of moving a Dirac point mass from $\xi_a$ to $\xi_b$ amounts to $\|\xi_a-\xi_b\|$.
\begin{equation*}
    d_w(\mathbb{P}_a,\mathbb{P}_b) = \inf_{\mathbb{P}\in\mathcal{P}(\mathbb{P}_a,\mathbb{P}_b)} \int \|\xi_a-\xi_b\|d \mathbb{P}(\xi_a,\xi_b),
\end{equation*}
where %$\xi_a\sim\mathbb{P}_a$, $\xi_b\sim\mathbb{P}_b$ and 
$\mathcal{P}(\mathbb{P}_a,\mathbb{P}_b)$ represents the set of all distributions with margins $\mathbb{P}_a$ and $\mathbb{P}_b$.
\end{definition}
% \begin{definition}
%     For the set $\mathcal{\bar{A}} = \{\mathbf{x}\in \mathbb{R}^n|\mathbf{A}\mathbf{x}\geq\bm{\beta}\mathbf{\xi}+b\}$ with $\mathbf{A}\in \mathbb{R}^{m\times n}$, $\bm{\beta}\in\mathbb{R}^{m\times l}$ and $b\in\mathbb{R}^m$, the distance between one realization $\hat{\xi}_k$ and the set is defined as \cite{MI-GNEP}:
% \begin{equation}\label{dist}
% \text{dist}\left(\hat{\xi}_k,\mathcal{\bar{A}}(\mathbf{x}) \right) = \min_{j\in[m]}\frac{(\bm{\beta}_j\hat{\xi}_k+b_j-\mathbf{A}_j\mathbf{x})^+}{\|\bm{\beta}_j\|_*},
% \end{equation}
% where $\|\cdot\|_*$ is the dual norm, $\mathbf{A}_j$ and $\bm{\beta}_j$ are the $j$th row of matrix $\mathbf{A}$ and $\bm{\beta}$ respectively, and $b_j$ is the $j$th element of $b$.
% \end{definition}
\subsection{Problem Statement}
We consider a game with $I$ agents indexed by $i\in[I]$ with decision variables $x_i\in\mathbb{R}^{n_i}$ for agent $i$. Moreover, we define the collective strategy profile as $\mathbf{x}=[x_1, x_2,\cdots,x_I]^\top\in \mathbb{R}^n$ with $n = \sum_{i=1}^I n_i$ and define the rivals' strategy profile for agent $i$ as $\mathbf{x}_{-i}=[x_1, x_2,\cdots,x_{i-1},x_{i+1},\cdots,x_I]^\top\in \mathbb{R}^{n_{-i}}$ with $n_{-i} = n-n_i$. For agent $i$, it is characterized by a local strategy set $\mathcal{S}_i\in\mathbb{R}^{n_i}$ and a cost function $J_i(x_i,\mathbf{x}_{-i}):\mathbb{R}^{n_i}\times\mathbb{R}^{n_{-i}}\rightarrow\mathbb{R}$ that may depend on both its own strategy $x_i$ and its rivals' strategies $\mathbf{x}_{-i}$. Each agent $i\in [I]$ aims to minimize its cost $J_i(x_i,\mathbf{x}_{-i})$ by choosing a strategy $x_i$ in its local strategy set $\mathcal{S}_i$. For any given $\mathbf{x}\in\mathbb{R}^n$, we define:
\begin{equation*}
    \mathcal{S}(\mathbf{x})=\prod_{i=1}^I\mathcal{S}_i
    = \{
    \mathbf{y}\in\mathbb{R}^n\left|
    y_i\in\mathcal{S}_i,\quad \forall i\in[I]
    \right
    \}.
\end{equation*}
To facilitate the illustration, we also give the following preliminary assumption throughout the paper.
\begin{assumption}\label{ass:S}
    The local strategy set $\mathcal{S}_i$ is nonempty, compact, and convex for all $i\in [I]$.
\end{assumption}

\emph{This paper aims to solve the following GNEP under shared DRCC and in particular develop an efficient method to compute the equilibrium:}
% a finite number of coupling constraints shared by all agents in the form of joint distributional robust chance constraints:
% \begin{equation} 
%     \inf_{\mathbb{P}\in \mathcal{P}}\mathbb{P}\left[ 
%     \mathbf{A}\mathbf{x}<\mathbf{b}(\mathbf{\xi})
%     \right]\geq 1-\epsilon,
% \end{equation}
% where $\mathcal{P}$ is the ambiguity set including a whole family of possible distributions of $\mathbf{\xi}$. $\mathbf{A}\in \mathbb{R}^{m\times n}$, $\mathbf{b}(\cdot):\mathbb{R}^l\rightarrow \mathbb{R}^m$ is an affine function $\mathbf{b}(\mathbf{\xi}) = \bm{\beta}\mathbf{\xi}+b$ with  $\bm{\beta}\in \mathbb{R}^{m\times l}$ and $b\in\mathbb{R}^m$. $\mathbf{\xi}\in\mathbb{R}^l$ is the uncertainty under some certain probability distribution $\mathbb{P}$. $\epsilon\in(0,1)$ is the prescribed violation rate of the constraint. By describing the constraint in the joint chance constraint form, it is guaranteed that the possibility of any one of the $m$ constraints being violated is less than $\epsilon$.
\begin{equation}\label{eq:problem}
i\in [I]:
\left\{
    \begin{aligned}
    \min&_{x_i}\quad J_i(x_i,\mathbf{x}_{-i})\\
    \text{s.t.}&\quad\inf_{\mathbb{P}\in \mathcal{P}}\mathbb{P}\left[ 
    \mathbf{A}\mathbf{x}<\mathbf{b}(\mathbf{\xi})
    \right]\geq 1-\epsilon,
    \end{aligned}
    \right.
\end{equation}
where $\mathbf{\xi}\in\mathbb{R}^l$ is the uncertainty under some certain probability distribution $\mathbb{P}$, $\mathcal{P}$ is the ambiguity set including a whole family of possible distributions of $\mathbf{\xi}$. $\mathbf{A}\in \mathbb{R}^{m\times n}$ and $\mathbf{b}(\cdot):\mathbb{R}^l\rightarrow \mathbb{R}^m$ is an affine function $\mathbf{b}(\mathbf{\xi}) = \bm{\beta}\mathbf{\xi}+b$ with  $\bm{\beta}\in \mathbb{R}^{m\times l}$ and $b\in\mathbb{R}^m$. $\epsilon\in(0,1)$ is the prescribed violation rate of the constraint. By describing the constraint in the joint chance constraint form, it is guaranteed that the possibility of any one of the $m$ constraints being violated is less than $\epsilon$.
In this case, the strategy set of agent $i\in [I]$ is given by
\begin{equation*}
    \mathcal{X}_i(\mathbf{x}_{-i})= \mathcal{S}_i\!\cap\!
    \left\{
    x_i\!\in\!\mathbb{R}^{n_i}\!\left|
    \inf_{\mathbb{P}\in \mathcal{P}}\mathbb{P}\left[ 
    \mathbf{A}\mathbf{x}<\mathbf{b}(\mathbf{\xi})
    \right]\geq 1-\epsilon
    \right.
    \right
    \},
\end{equation*}
which depends on the strategies of the other agents. The coupled collective strategy set is then defined as
$
    \mathcal{X}(\mathbf{x})=    \prod_{i=1}^I\mathcal{X}_i(\mathbf{x}_{-i}).
$
The GNEP defined above can be represented by $\aleph 
= (I, (\mathcal{X}_i(\mathbf{x}_{-i}))_{i\in[I]},(J_i(x_i,\mathbf{x}_{-i}))_{i\in[I]})$. For GNEP $\aleph$, GNE is defined as follows.
\begin{definition}
    A collective strategy $\mathbf{x}^*$ is a generalized Nash equilibrium if for all $i\in [I]$
\begin{align*}
    J_i(x^*_i,\mathbf{x}^*_{-i})\leq J_i(x_i,\mathbf{x}^*_{-i}), \quad\forall x_i\in \mathcal{X}_{i}(\mathbf{x}^*_{-i}).
\end{align*}
\end{definition}
That is to say, a GNE is a collective strategy $\mathbf{x}^*$ such that no agent can improve its aim by changing its strategy $x_i^*$ to another feasible one $y_i\in\mathcal{X}_{i}(\mathbf{x}^*_{-i})$.
For the sake of further discussion on computing the GNE, let us recall the NI function below. 
\begin{definition}[\cite{NIfunc}]
    For any $\mathbf{x}, \mathbf{y}\in \mathbb{R}^n$, the NI-function is defined as:
    \begin{equation*}\label{NI}
        \Psi(\mathbf{x}, \mathbf{y}) = \sum_{i\in[I]}\left[J_i(x_i,\mathbf{x}_{-i})-J_i(y_i,\mathbf{x}_{-i})\right].
    \end{equation*}
\end{definition}
Let
\begin{equation*}\label{V}
    \Hat{V}(\mathbf{x}) = \sup_{\mathbf{y}\in\mathcal{X}(\textbf{x})} \Psi(\mathbf{x},\mathbf{y}).
\end{equation*}
% $\Hat{V}(\mathbf{x}) = \sup_{\mathbf{y}\in\mathcal{X}(\textbf{x})} \Psi(\mathbf{x},\mathbf{y})$. Notice that $\Hat{V}(\mathbf{x})\geq 0$ 
for all $\mathbf{x}\in \mathcal{X}$, then the following results hold.
\begin{theorem}[\cite{FaccGNE}]
    For a GNEP $\aleph$, the following statements are equivalent.
    \begin{enumerate}
        \item $\mathbf{x}^*$ is a generalized Nash equilibrium for $\aleph$.
        \item $\mathbf{x}^*\in \mathcal{X}(\mathbf{x}^*)$ and $\Hat{V}(\mathbf{x}^*) = 0$.
        \item $\mathbf{x}^*$ is an optimal solution of $\inf_{\mathbf{x}\in \mathcal{X}} \Hat{V}(\mathbf{x})$ with value zero.
    \end{enumerate}
\end{theorem}
%------------------------
%The following definitions are instrumental for the upcoming discussion.
\begin{definition}
    For a GNEP $\aleph$, the domain of the strategy set $\mathcal{X}_i(\mathbf{x}_{-i})$ is
\begin{equation*}\label{dom}
        \text{dom}\mathcal{X}_i = \left\{ \mathbf{x}_{-i}\in \mathbb{R}^{n_{-i}}\left|
        \mathcal{X}_i(\mathbf{x}_{-i})\neq \emptyset
        \right.
        \right\}.
\end{equation*}
\end{definition}
According to $\text{dom}\mathcal{X}_i$, the refined domain $\text{rdom}\mathcal{X}_i$, which can be seen as the projection of $\text{dom}\mathcal{X}_i$ onto $\prod_{j \neq i,j\in [I]}\mathcal{X}_{j}$, considers whether rivals' strategies $\mathbf{x}_{-i}$ are in their own strategy sets.
\begin{definition}
    For a GNEP $\aleph$, the refined domain of the strategy set $\mathcal{X}_i(\mathbf{x}_{-i})$ is defined by
    \begin{multline*}\label{rdom}
        \text{rdom}\mathcal{X}_i = \left\{ \mathbf{x}_{-i}\in \mathbb{R}^{n_{-i}}|
        \exists x_i\in\mathbb{R}^{n_i}: \right.\\ \left.(x_i, \mathbf{x}_{-i})\in \mathcal{X}(x_i, \mathbf{x}_{-i})
        %\right.
        \right\}.
    \end{multline*}
\end{definition}

\section{Distributionally Robust GNE over the Wasserstein Ball}\label{sec:DRCC}
This section discusses the GNE of problem \eqref{eq:problem} by leveraging the concept of the Wasserstein ball, which serves as the ambiguity set defined by the available uncertainty samples.

% Classic chance constraints require full information on the probability distribution of the uncertainty, which is impossible in many cases. Nevertheless, we can obtain samples of uncertainties based on historical data. Therefore, DRCCs are preferred in many studies. The constraint is required to be satisfied with high assurance over a whole family of probability distributions defined over the ambiguity set instead of a specific distribution in DRCC. 
%We can define the ambiguity set based on historical data. 
% Specifically, the ambiguity set designed through the Wasserstein ball centered around the empirical distribution is explored based on historical data.

%\begin{assumption}\label{data}
    The historical  data of the uncertainty $\xi$ is a set $\mathcal{U}=\{\hat{\xi}_1, \hat{\xi}_2,\cdots,\hat{\xi}_K\}$ where  $\hat{\xi}_k$'s are sampled i.i.d. at random following an unknown true distribution of $\xi$.
%\end{assumption}
% Assume we obtain $K$ iid realizations of uncertainty $\mathbf{\xi}$ from historical data, represented by a set $\mathcal{U}=\{\hat{\xi}_1, \hat{\xi}_2,\cdots,\hat{\xi}_K\}$.
The empirical distribution is $\mathbb{\hat{P}}_e= \frac{1}{K}\sum_{k = 1}^K\delta_{\hat{\xi}_k}$. 

%\begin{assumption}\label{wb}
Consider the ambiguity set defined by a Wasserstein ball centered around $\mathbb{\hat{P}}_e$ with radius $\theta>0$%\footnote{\color{blue}The selection of $\theta$ can be guided by [Theorem 3.4] in \cite{dw_theta}.}
\begin{equation*}
    \mathcal{P}_{\theta}(\mathbb{\hat{P}}_e) = \{\mathbb{P}\mid d_w(\mathbb{\hat{P}}_e,\mathbb{P})\leq \theta\}.
\end{equation*}

To this end, we address the following GNEP with the DRCC over the Wasserstein ball:
\begin{equation}\label{dw}
i\in [I]:
\left\{
    \begin{aligned}
    \min&_{x_i\in\mathcal{S}_i}\quad J_i(x_i,\mathbf{x}_{-i})\\
    &\text{s.t.}\quad    \inf_{\mathbb{P}\in \mathcal{P}_{\theta}(\mathbb{\hat{P}}_e)}\mathbb{P}
    \left[ 
    \mathbf{A}\mathbf{x}<\mathbf{b}(\mathbf{\xi})
    \right]\geq 1-\epsilon.
    \end{aligned}
    \right.
\end{equation}
% \begin{equation}\label{dw}
%     \inf_{\mathbb{P}\in \mathcal{P}_{\theta}(\mathbb{\hat{P}}_e)}\mathbb{P}
%     \left[ 
%     \mathbf{A}\mathbf{x}<\mathbf{b}(\mathbf{\xi})
%     \right]\geq 1-\epsilon.
% \end{equation}
As illustrated in \cite{dWRefo}, the DRCC \eqref{dw} defined over the Wasserstein ball can be interpreted as the idea that the minimum cost required to transport $\epsilon K$ samples from $\hat{\xi}_k\in \mathcal{U}$ to the unsafe set $\mathcal{\bar{A}} = \{\mathbf{x}\in \mathbb{R}^n|\mathbf{A}\mathbf{x}\geq\bm{\beta}\mathbf{\xi}+b\}$
should be no less than $\theta K$. 
As such, the DRCC in \eqref{dw} is equivalent to:
\begin{equation}\label{dist2}
   \sum_{i= 1}^{\epsilon K}\text{dist}\left(\hat{\xi}_{k^\prime(i)},\mathcal{\bar{A}}(\mathbf{x}) \right)\geq \theta K,
\end{equation}
%  where the distance between one realization $\hat{\xi}_{k^\prime(i)}$ and the set $\mathcal{\bar{A}}$ is defined as \cite{MI-GNEP}
% \begin{equation}\label{dist}
% \text{dist}\left(\hat{\xi}_{k^\prime(i)},\mathcal{\bar{A}}(\mathbf{x}) \right) = \min_{j\in[m]}\frac{(\bm{\beta}_j\hat{\xi}_{k^\prime(i)}+b_j-\mathbf{A}_j\mathbf{x})^+}{\|\bm{\beta}_j\|_*} 
% \end{equation}
% with $\|\cdot\|_*$ being the dual norm, $\mathbf{A}_j$ and $\bm{\beta}_j$ are the $j$th row of matrix $\mathbf{A}$ and $\bm{\beta}$ respectively, and $b_j$ is the $j$th element of $b$.
with the left-hand side defined as
$$\sum_{i= 1}^{\lfloor\epsilon K\rfloor}\!\text{dist}\!\left(\hat{\xi}_{k^\prime(i)},\mathcal{\bar{A}}(\mathbf{x}) \right)\!+\!(\epsilon K\! -\! \lfloor\epsilon K\rfloor)\text{dist}\!\left(\hat{\xi}_{k^\prime(\lfloor\epsilon K\rfloor+1)},\mathcal{\bar{A}}(\mathbf{x}) \right)$$
and $k^\prime(i)$ being the index that the samples are in the order of
\begin{equation*}
\begin{aligned}    \text{dist}\left(\hat{\xi}_{k^\prime(1)},\mathcal{\bar{A}}(\mathbf{x}) \right)\leq   \text{dist}\left(\hat{\xi}_{k^\prime(2)},\mathcal{\bar{A}}(\mathbf{x}) \right)\leq 
\cdots\\
\leq
\text{dist}\left(\hat{\xi}_{k^\prime(K)},\mathcal{\bar{A}}(\mathbf{x}) \right).
\end{aligned}
\end{equation*}
Summing the $\epsilon K $ smallest distances in \eqref{dist2} can be cast as a LP problem \cite{dWRefo}. 
% The calculation of the left-hand side of \eqref{dist2} can be formulated as an optimization program. 
By strong linear programming duality, the calculation can be transformed into its dual problem as:
\begin{equation}\label{dual}
    \begin{aligned}
        \max_{\mathbf{s},\tau}\quad&\epsilon K\tau-\mathbf{e}_K^\top\mathbf{s}\\
        \text{s.t}\quad &\text{dist}\left(\hat{\xi}_k,\mathcal{\bar{A}}(\mathbf{x}) \right)\geq \tau-s_k,\quad\forall k\in[K]\\
        &\mathbf{s}\geq 0,
    \end{aligned}
\end{equation}
where $s_k$ is the $k$th element in vector $\mathbf{s}\in\mathbb{R}^K$. The optimal value of $\epsilon K\tau-\mathbf{e}_K^\top\mathbf{s}$ equals the left-hand side of \eqref{dist2}.
\begin{lemma}[\cite{dWRefo}]
The distance of a point ${\xi}_{k^\prime(i)}$ to $\mathcal{\bar{A}}$ can be determined by 
\begin{equation}\label{dist}
\text{dist}\left(\hat{\xi}_{k^\prime(i)},\mathcal{\bar{A}}(\mathbf{x}) \right) = \min_{j\in[m]}\frac{(\bm{\beta}_j\hat{\xi}_{k^\prime(i)}+b_j-\mathbf{A}_j\mathbf{x})^+}{\|\bm{\beta}_j\|_*} 
\end{equation}
with $\|\cdot\|_*$ the dual norm, $\mathbf{A}_j$ and $\bm{\beta}_j$ are the $j$th row of matrix $\mathbf{A}$ and $\bm{\beta}$ respectively, and $b_j$ is the $j$th element of $b$.
\end{lemma}

The reformulation process indicates that an optimization problem with DRCC can be reconstructed into two levels with \eqref{dual} being the lower-level problem.
If considering the Nash game framework with shared DRCC, the problem presented in \eqref{dw} is reconstructed as a bilevel game-theoretic framework, wherein all agents participate in a Nash game at the leader level while sharing a unified follower system defined in \eqref{dual}.
Next, we aim to simplify it into a structure that is more accessible, as detailed below.

%\subsection{DRCC Reformulation}
%Due to its probabilistic nature, DRCC is challenging to handle directly. Next, we reformulate DRCC within the GNEP framework.
\begin{theorem}
    The solution of GNEP \eqref{dw} with a shared DRCC is equivalent to the solution of the following GNEP problem in the deterministic form:
\begin{subequations}\label{game_MI}
\allowdisplaybreaks
\begin{align}
&i\in [I]:\nonumber\\
&\left\{
    \begin{aligned}
    &\min_{x_i\in\mathcal{S}_i}\quad J_i(x_i,\mathbf{x}_{-i})\\    &s.t.\quad(\overline{\bm{\beta}}\hat{\xi}+\overline{b}-\sum_{i\in[I]}\overline{\mathbf{A}}_{n_i}x_i)+\overline{M}\mathbf{q}\geq\tau^{\prime}\mathbf{e}_{mK}-\overline{E}\mathbf{s}^{\prime}
    \end{aligned}
    \right.\\
&i = I+1:\nonumber\\
&\left\{
    \begin{aligned}
    &\min_{x_{(I+1)}}\quad  0\\
    &s.t.\quad
    x_{(I+1)}\in \mathcal{B}
    \end{aligned}
    \right.
    \end{align}
\end{subequations}
where $x_{(I+1)}=\quad[\tau^{\prime}, \mathbf{s}^{\prime\top}, \mathbf{q}^\top]^\top$,
\begin{equation*}
\mathcal{B}   
\!=\!
\left\{x_{(I+1)}\! \in \!\mathbb{R}^{2K+1}
\!\mid\!
h_i(x_{(I+1)})\!\text{ holds for all }
\!i\!=\!1,\ldots,5
\right\},
\end{equation*}
%\begin{subnumcases}{\mathcal{B}=}\label{SetI+1}
%    \quad\epsilon K\tau^{\prime}-\mathbf{e}_K^\top \mathbf{s}^{\prime} \geq \theta K||\bm{\beta}||_*,\label{C1}\\    (\overline{\bm{\beta}}\hat{\xi}+\overline{b}-\sum_{i\in[I]}\overline{\mathbf{A}}_{n_i}x_i)+\overline{M}\mathbf{q}\geq\tau^{\prime}\mathbf{e}_{mK}-\overline{E}\mathbf{s}^{\prime},\label{C2}\\
 %   M(\mathbf{e}_K-\mathbf{q})\geq \tau^{\prime}\mathbf{e}_K-\mathbf{s}^{\prime},\label{C3} \\
%    \mathbf{s}^{\prime}\geq 0,\label{C4}\\
%    \mathbf{q}\in\{0,1\}^K,\label{C5}
%\end{subnumcases}
with 
\begin{subnumcases}{}
    h_1:\quad\epsilon K\tau^{\prime}-\mathbf{e}_K^\top \mathbf{s}^{\prime} \geq \theta K||\bm{\beta}||_*,\label{C1}\\ 
    h_2: (\overline{\bm{\beta}}\hat{\xi}+\overline{b}-\!\sum_{i\in[I]}\overline{\mathbf{A}}_{n_i}x_i)+\overline{M}\mathbf{q}\geq\!\tau^{\prime}\mathbf{e}_{mK}-\overline{E}\mathbf{s}^{\prime},\label{C2}\\
    h_3: M(\mathbf{e}_K-\mathbf{q})\geq \tau^{\prime}\mathbf{e}_K-\mathbf{s}^{\prime},\label{C3} \\
    h_4: \mathbf{s}^{\prime}\geq 0,\label{C4}\\
    h_5: \mathbf{q}\in\{0,1\}^K,\label{C5}
\end{subnumcases}
% \begin{subequations}
% \begin{empheq}[left=\empheqlbrace]{align}
%     h_1:&\quad\epsilon K\tau^{\prime}-\mathbf{e}_K^\top \mathbf{s}^{\prime} \geq \theta K||\bm{\beta}||_*,\label{C1}\\   
%     h_2:& (\overline{\bm{\beta}}\hat{\xi}+\overline{b}-\sum_{i\in[I]}\overline{\mathbf{A}}_{n_i}x_i)+\overline{M}\mathbf{q}\geq\tau^{\prime}\mathbf{e}_{mK}-\overline{E}\mathbf{s}^{\prime},\label{C2}\\
%     h_3:& M(\mathbf{e}_K-\mathbf{q})\geq \tau^{\prime}\mathbf{e}_K-\mathbf{s}^{\prime},\label{C3} \\
%     h_4:& \mathbf{s}^{\prime}\geq 0,\label{C4}\\
%     h_5:& \mathbf{q}\in\{0,1\}^K,\label{C5}
% \end{empheq}
% \end{subequations}
% \begin{subequations}\label{SetI+1}
%     \begin{align}
%     \mathcal{B}=\Big\{\\
%     &
%     \quad\epsilon K\tau-\mathbf{e}_K^\top \mathbf{s} \geq \theta K||\bm{\beta}||_*,\label{C1}\\
%     &(\overline{\bm{\beta}}\hat{\xi}+\overline{b}-\sum_{i\in[I]}\overline{\mathbf{A}}_{n_i}x_i)+\overline{M}\mathbf{q}\geq\tau\mathbf{e}_{mK}-\overline{E}\mathbf{s},\label{C2}\\
%     &M(\mathbf{e}_K-\mathbf{q})\geq \tau\mathbf{e}_K-\mathbf{s}, \mathbf{q}\in\{0,1\}^K\Big\}.\label{C3}
%     \end{align}
% \end{subequations}
where $\overline{\bm{\beta}} = \mathbb{I}_K\otimes\bm{\beta}$, $\overline{b} = \mathbf{e}_K\otimes b$, $\overline{\mathbf{A}}_{n_i} = \mathbf{e}_K\otimes \mathbf{A}_{n_i}$ with $\mathbf{A}_{n_i}$ being the columns of $\mathbf{A}$ corresponding to $x_i$, $M\in{\mathbb{R}^n}_{> 0}$ is a suitable large positive constant, $\overline{M} = (\mathbb{I}_K\otimes\mathbf{e}_m)M$. $\tau^{\prime}\in\mathbb{R}$, $\mathbf{s}^{\prime}\in\mathbb{R}^K$, and $\mathbf{q}$ are decision variables of the (I+1)th agent. $\mathbf{q}$ is the binary auxiliary vector with $q_k$ being the $k$th element of $\mathbf{q}$, and $\overline{E} = (\mathbb{I}_K\otimes\mathbf{e}_m)$.
\end{theorem}
\begin{proof}
Defining agents $i\in[I]$ as leaders and the $(I+1)$th agent as the follower, the bilevel game is expressed as:
\begin{subequations}\label{game1}
\allowdisplaybreaks
\begin{align}
&\text{Leaders }i\in [I]: \nonumber\\
&\left\{
    \begin{aligned}
    \min_{x_i\in\mathcal{S}_i}\quad &J_i(x_i,\mathbf{x}_{-i})\\    \text{s.t.}\quad&\epsilon K\tau-\mathbf{e}_K^\top\mathbf{s}\geq\theta K,\\
    &\text{dist}\left(\hat{\xi}_k,\mathcal{\bar{A}}(\mathbf{x}) \right)\geq \tau-s_k,\quad\forall k\in[K]
    \end{aligned}
    \right.\\
    &\text{Follower }i = I+1:\nonumber\\
&\left\{
    \begin{aligned}
    \max_{\tau,\mathbf{s}}\quad &\epsilon K\tau-\mathbf{e}_K^\top\mathbf{s}\\
    \text{s.t}\quad &\text{dist}\left(\hat{\xi}_k,\mathcal{\bar{A}}(\mathbf{x}) \right)\geq \tau-s_k,\quad\forall k\in[K]\\
    &\mathbf{s}\geq 0,
    \end{aligned}
    \right.
    \end{align}
\end{subequations}
with the leaders performing a Nash game between each other. It follows that \eqref{game1} is equivalent to:
% Leaders $i\in [I]$:
% \begin{subequations}
%     \begin{align}
%     \min_{x_i\in\mathcal{S}_i}\quad &J_i(x_i,\mathbf{x}_{-i})\\    
%     \text{s.t.}\quad
%     &\text{dist}\left(\hat{\xi}_k,\mathcal{\bar{A}}(\mathbf{x}) \right)\geq \tau-s_k,\quad\forall k\in[K].
%     \end{align}
% \end{subequations}
% Follower $i = I+1$:
% \begin{subequations}
%     \begin{align}
%     \max_{\tau,\mathbf{s}}\quad &0\\
%     \text{s.t}\quad&\epsilon K\tau-\mathbf{e}_K^\top\mathbf{s}\geq \theta K    \\&\text{dist}\left(\hat{\xi}_k,\mathcal{\bar{A}}(\mathbf{x}) \right)\geq \tau-s_k,\quad\forall k\in[K]\\
%     &\mathbf{s}\geq 0.
%     \end{align}
% \end{subequations}
%------------------------------------
\begin{subequations}\label{game2}
\allowdisplaybreaks
\begin{align}
&\text{Leaders }i\in [I]:\nonumber\\
&\left\{
    \begin{aligned}
    \min_{x_i\in\mathcal{S}_i}\quad &J_i(x_i,\mathbf{x}_{-i})\\    
    \text{s.t.}\quad
    &\text{dist}\left(\hat{\xi}_k,\mathcal{\bar{A}}(\mathbf{x}) \right)\geq \tau-s_k,\quad\forall k\in[K]
    \end{aligned}
    \right.\\
    &\text{Follower }i = I+1:\nonumber\\
    &\left\{
    \begin{aligned}\label{C0I+1}
    \max_{\tau,\mathbf{s}}\quad &0\\
    \text{s.t}\quad&\epsilon K\tau-\mathbf{e}_K^\top\mathbf{s}\geq \theta K    \\&\text{dist}\left(\hat{\xi}_k,\mathcal{\bar{A}}(\mathbf{x}) \right)\geq \tau-s_k,\quad\forall k\in[K]\\
    &\mathbf{s}\geq 0.
    \end{aligned}
    \right.
    \end{align}
\end{subequations}
The equivalence of \eqref{game1} and \eqref{game2} follows from the mutual feasibility of the two problems. Suppose we have an optimal $\mathbf{x}, \tau, \mathbf{s}$ for \eqref{game2}, then $\mathbf{x}$ is also feasible for \eqref{game1}. As the objective of agent $i\in[I]$ is irrelevant to $\tau$ and $\mathbf{s}$, they also have the same optimal value. Conversely, suppose there is an optimal $\mathbf{x}$ for \eqref{game2}. As $\mathbf{x}$ is feasible for \eqref{game2}, there must exist feasible $\tau, \mathbf{s}$ satisfying the constraints of the additional agent. Then $\mathbf{x}, \tau, \mathbf{s}$ will also be feasible for \eqref{game1} with the same optimal value for agent $i\in[I]$.

As the objective of the follower remains constant at $0$, which implies that the follower is indifferent about changing its strategy regardless of any strategy of the leaders. The strategy of the leader is chosen to minimize its own cost given that the follower has no intention to deviate from any action. This aligns precisely with the condition for the Nash game, where any agent has no intention to change its strategy at equilibrium. Therefore, we can reformulate \eqref{game2} as a Nash game with $I+1$ agents.

After reformulating \eqref{dw} as the Nash game with $I+1$ agents, we plug \eqref{dist} into the second group of constraints, leading to
% \begin{equation}
% \begin{aligned}
% \min_{j\in[m]}&\frac{(\bm{\beta}_j\hat{\xi}_k+b_j-\mathbf{A}_j\mathbf{x})^+}{\|\bm{\beta}_j\|_*}\geq\tau-s_k,\quad\forall j\in [m],\forall k\in [K]\\
% =&\frac{(\bm{\beta}\hat{\xi}_k+b-\mathbf{A}\mathbf{x})^+}{\|\bm{\beta}\|_*}\geq\tau-s_k, \quad\forall k\in [K].
% \end{aligned}
% \end{equation}
\begin{eqnarray*}
&\displaystyle \min_{j\in[m]}\frac{(\bm{\beta}_j\hat{\xi}_k+b_j-\mathbf{A}_j\mathbf{x})^+}{\|\bm{\beta}_j\|_*}\geq\tau-s_k,\quad\forall k\in [K] \nonumber \\ 
&\displaystyle \Leftrightarrow \frac{(\bm{\beta}_j\hat{\xi}_k+b_j-\mathbf{A}_j\mathbf{x})^+}{\|\bm{\beta}_j\|_*}\geq\tau-s_k, \quad\forall j\in [m], \forall k\in [K]. 
\end{eqnarray*}

Notice that the maximum operator $(\cdot)^+$ introduces nonconvexity to the problem. Following the methods in \cite{dWRefo}, to deal with the nonconvexity, we first eliminate the molecule $||\bm{\beta}||_*$ by substituting $\tau$ and $\mathbf{s}$ by $\tau^{\prime}$ and $\mathbf{s}^{\prime}$, respectively. Then, we utilize the big-M method to handle the $(\cdot)^+$ operator. The constraints in \eqref{C0I+1} can be expressed as \eqref{C1} - \eqref{C5}.
% \begin{equation*}\label{bigM}
%     \begin{aligned}
%         \epsilon K\tau^{\prime}-\mathbf{e}_K^\top \mathbf{s}^{\prime} &\geq \theta K||\bm{\beta}||_*\\
%         (\bm{\beta}\hat{\xi}_k+b-\mathbf{A}\mathbf{x})+Mq_k&\geq\tau^{\prime}-s^{\prime}_k,\quad\forall k\in [K]\\
%         M(1-q_k)&\geq \tau^{\prime}-s^{\prime}_k,\quad\forall k\in [K]\\
%         \mathbf{q}\in\{0,1\}^K,& \quad
%         \mathbf{s}^{\prime}\geq 0.
%     \end{aligned}
% \end{equation*}
As proved in \cite{DRCCgame}, decision variables $\tau^{\prime}$, $\mathbf{s}^{\prime}$, and $\mathbf{q}$ are all bounded. Therefore, there exists a big enough but finite M for the reformulation to be exact. 
% Besides, the boundedness of $\tau$, $\mathbf{s}$, and $\mathbf{q}$ guarantees that the reformulation will not affect the optimal value of the problem. 
Finally, \eqref{game_MI} can be obtained by applying \eqref{C1}-\eqref{C5} to \eqref{game2}.
\end{proof}

Now, the original GNEP has been reformulated in a deterministic form. However, the binary auxiliary variable $\mathbf{q}\in\{0,1\}^K$ introduced by the big-M method induces discontinuities, posing significant challenges in solving the GNEP. These challenges are addressed in the next Theorem. %Hence, in the next section, we will discuss the way to deal with the binary variables by convexifying the corresponding strategy set.

\begin{theorem}\label{EQ}
    For GNEP \eqref{game_MI}, 
    % if the convex envelope $J^c_i(y_i,\mathbf{x}_{-i})$ of the cost function satisfies $J^c_i(y_i,\mathbf{x}_{-i}) = J_i(y_i,\mathbf{x}_{-i})$ over $\mathbf{x}_{-i}\in \text{rdom}\mathcal{X}_i$ for all $i\in[I]$, 
    let 
        \begin{equation}\label{Vc}
        \hat{V}^c(\mathbf{x}) = \sup_{\mathbf{y}\in \mathcal{X}^c(\mathbf{x})}\sum_{i = 1}^{I}\left[J^c_i(\mathbf{x})-J^c_i(y_i,\mathbf{x}_{-i})\right],
        \end{equation}
    then the following statements are equivalent.
    \begin{enumerate}
        \item $\mathbf{x}^*$ is a generalized Nash equilibrium for \eqref{game_MI}.
        \item $\mathbf{x}^*$ is a generalized Nash equilibrium for the game $\aleph^c = \left(I+1,\mathcal{X}^c(\mathbf{x}),(J^c_i)_{i\in[I+1]}\right)$, where 
       $
        \mathcal{X}^c(\mathbf{x}) =    \prod_{i=1}^I\mathcal{X}_i(\mathbf{x}_{-i})\cap \mathcal{X}^c_{I+1}(\mathbf{x}_{-(I+1)}),
    $
 $\mathcal{X}^c_{I+1}(\mathbf{x}_{-(I+1)})$ is the canonical relaxation of the original strategy set:
    \begin{equation}\label{ConvX}
    \begin{aligned}
       \mathcal{B}^c\!=\!\Big\{x_{(I+1)}\in\mathbb{R}^{2K+1}\Big|\eqref{C1}\! - \!\eqref{C4},
    %     \quad\epsilon K\tau-\mathbf{e}_K^\top \mathbf{s} \geq \theta K||\bm{\beta}||_*\\   &(\overline{\bm{\beta}}\hat{\xi}+\overline{b}-\sum_{i\in[I]}\overline{\mathbf{A}}_{n_i}x_i)+\overline{M}\mathbf{q}\geq\tau\mathbf{e}_{mK}-\overline{E}\mathbf{s}\\
    % &M(\mathbf{e}_K-\mathbf{q})\geq \tau\mathbf{e}_K-\mathbf{s}\\
    \mathbf{q}\in[0,1]^K\Big\}.
    \end{aligned}
    \end{equation}
    \item $\mathbf{x}^*$ is an optimal solution of
        \begin{equation}\label{SolVc}
            \inf_{\mathbf{x}\in \mathcal{X}(\mathbf{x})} \hat{V}^c(\mathbf{x})
        \end{equation}
    with value zero.
    % \item  $\mathbf{x}^* \in\arg \inf_{\mathbf{x}\in \mathcal{X}(\mathbf{x})} \hat{V}^c(\mathbf{x})$ and $\hat{V}^c(\mathbf{x}^*) = 0$.   
    \end{enumerate}      
\end{theorem}
\begin{proof}   
    According to Assumption \ref{ass:S} and the linearity of \eqref{C2}, the strategy set of the $(I+1)$th agent-whose decision variables include the integer variable $\mathbf{q}$ is the only non-convex strategy set among all agents. Therefore, the convexified overall strategy set can be expressed as:
    \begin{equation*}
        \mathcal{X}^c(\mathbf{x}) =    \prod_{i=1}^I\mathcal{X}_i(\mathbf{x}_{-i})\cap \mathcal{X}^c_{I+1}(\mathbf{x}_{-(I+1)})
    \end{equation*}
    with strategy sets for agent $i\in[I]$ remaining unchanged and $\mathcal{X}^c_{I+1}(\mathbf{x}_{-(I+1)})$ being the convex hull of $(I+1)$th agent.
%     We first define the strategy set of $(I+1)$th agent as set $\mathcal{B}$:
%     \begin{equation}\label{SetI+1}
%     \begin{aligned}
%     \mathcal{B}=\Big\{&x_{(I+1)}\in\mathbb{R}^{2K+1}\Big|
%     \quad\epsilon K\tau-\mathbf{e}_K^\top \mathbf{s} \geq \theta K||\bm{\beta}||_*,\\
%     &(\overline{\bm{\beta}}\hat{\xi}+\overline{b}-\sum_{i\in[I]}\overline{\mathbf{A}}_{n_i}x_i)+\overline{M}\mathbf{q}\geq\tau\mathbf{e}_{mK}-\overline{E}\mathbf{s},\\
%     &M(\mathbf{e}_K-\mathbf{q})\geq \tau\mathbf{e}_K-\mathbf{s}, \mathbf{q}\in\{0,1\}^K\Big\}.
%     \end{aligned}
% \end{equation}
    % We found $\mathcal{B}$ satisfies:
    %    \begin{equation}
    %     \mathcal{B} = \text{conv}\mathcal{B}\cap\mathbb{R}^{K+1}\times\mathbb{Z}^{K}.
    % \end{equation}
    % From the definition of the convex hull, we have $\mathcal{B}\subseteq\text{conv}\mathcal{B}$. 
    For the canonical relaxation $\mathcal{B}^c$ of the strategy set $\mathcal{B}$, the following relationship holds:
   % \begin{equation}\label{Bc}
   $
        \mathcal{B} = \mathcal{B}^c\cap\mathbb{R}^{K+1}\times\mathbb{Z}^{K}
   $ 
   % \end{equation}
    which further implies
    \begin{equation*}        \mathcal{B} = \mathcal{B}^c\cap\mathbb{R}^{K+1}\times\mathbb{Z}^{K}\subseteq \text{conv}\mathcal{B}\cap\mathbb{R}^{K+1}\times\mathbb{Z}^{K},
    \end{equation*}   
     due to $\mathcal{B}\subseteq\text{conv}\mathcal{B}$.
   As the canonical relaxation $\mathcal{B}^c$ is linear and convex, $\text{conv}\mathcal{B}$ is contained in $\mathcal{B}^c$ from the definition of convex hull. As a consequence, it is immediate to show that
   % \begin{equation*}
    %\text{conv}\mathcal{B}\cap\mathbb{R}^{K+1}\times\mathbb{Z}^{K}\subseteq\mathcal{B}^c\cap\mathbb{R}^{K+1}\times\mathbb{Z}^{K} = \mathcal{B}.
   % \end{equation*}
   % In conclusion, we have
    $
    %\begin{equation*}\label{HF}
        \mathcal{B}  =  \text{conv}\mathcal{B}\cap\mathbb{R}^{K+1}\times\mathbb{Z}^{K},
    %\end{equation*}
    $
    which implies that $\mathcal{B}$ is a hole-free set (see Definition~\ref{def:hole-free}).
This means the canonical relaxation $\mathcal{B}^c$ does not include additional integer points for any $\mathbf{x}_{-(I+1)}$ and coincides with the convex hull. Therefore, the convex hull of $(I+1)$th strategy set can be expressed as \eqref{ConvX}, and the convexified representation of problem \eqref{game_MI} can be represented by $\aleph^c = \left(I,\mathcal{X}^c(\mathbf{x}),(J^c_i)_{i\in[I]}\right)$. 
%--------
As illustrated in \cite[Theorem 2]{MI-GNEP}, to admit identical GNE with the original GNEP, $J^c_i(y_i,\mathbf{x}_{-i})$ for the convexified problem should satisfy $J^c_i(y_i,\mathbf{x}_{-i}) = J_i(y_i,\mathbf{x}_{-i})$ over $\mathbf{x}_{-i}\in \text{rdom}\mathcal{X}_i$ for all $i\in[I]$ at the equilibrium. Fulfilling this condition is generally challenging. In our formulation \eqref{game_MI}, the objective function associated with the nonconvex strategy set (strategy set of the $(I+1)$th agent) is zero, thereby inherently satisfying the required condition. This allows us to derive a convexified version of the original problem by simply taking the convex hull of the strategy set.
%--------
Therefore, the equivalence of three statements can be inferred from \cite[Theorem 2]{MI-GNEP}. Therefore, we can conclude Theorem \ref{EQ}.
\end{proof}
\begin{remark}
    It is worth mentioning that $\aleph^c$ is not equivalent to the fully continuous counterpart of GNEP $\aleph$ with $\mathbf{q}\in[0,1]^K$. As indicated in \eqref{SolVc}, the canonical extension is only applied to $\hat{V}^c(\mathbf{x})$. When it comes to the outer layer problem, we use the original strategy sets where $\mathbf{q}\in\{0,1\}^K$. This implies that in view of the $\text{rdom}\mathcal{X}_i$ for agent $i\in[I]$, $\mathbf{q}$ remains as a binary vector. %Accordingly, many analytical methods for convex GNEP cannot be applied directly to $\aleph^c$.
\end{remark}

\begin{remark}
    According to \cite{dw_theta}, the Wasserstein radius $\theta$ can be chosen based on the concentration bound. Specifically, the minimal Wasserstein radius $\theta$ that contains the true distribution with probability at least $1 - \epsilon$ satisfies:
    \[
    \theta \propto C\left(\frac{\text{log}(\epsilon^{-1})}{K}\right)^{1/\max\{m,2\}},
    \]
    where $m$ is the dimension of uncertainty, $a$ is the tail exponent, $C$ is a positive constant that depends explicitly on the distribution of the uncertainty and $m$. However, in practice, choosing the Wasserstein radius precisely is difficult, and some degree of manual tuning is often necessary.
\end{remark}
% {\blue
% \begin{remark}
%     Under mild regularity conditions that the uncertainty distribution is light-tailed (trivially satisfied when the support of the uncertainty is compact), the Wasserstein radius $\theta$ can be chosen based on the concentration bound given in [Theorem 3.4] of \cite{dw_theta}. Specifically, the minimal Wasserstein radius $\theta$ that contains the true distribution with probability at least $1 - \epsilon$ satisfies:
%     \[
%     \theta \propto C\left(\frac{\text{log}(\epsilon^{-1})}{K}\right)^{1/\max\{m,2\}},
%     \]
%     where $m$ is the dimension of uncertainty, $a$ is the tail exponent, $C$ is a positive constant that depends explicitly on the distribution of the uncertainty and $m$. However, in practice, choosing the Wasserstein radius precisely is difficult, and some degree of manual tuning is often necessary.
% \end{remark}
% }

%-------------------------------------------
\subsection{Solution Method}
From Theorem \ref{EQ}, finding the GNE relies on the solution to \eqref{SolVc}, which leads to the nested $\min-\max$ structure. 
% We didn't specify the form of the cost function, and the constraints previously. 
Generally speaking, this is a computationally intractable optimization problem. However, when $J_i(x_i,\mathbf{x}_{-i})$ and constraints of the convexified problem $\aleph^c$ meet certain conditions, the problem will become more manageable. 
% Therefore, in what follows, we will focus on a certain subclass of the problem with cost functions in the quadratic form and polyhedral individual strategy sets, then reformulate \eqref{SolVc} into a solvable form. 
\begin{assumption}\label{ass:Q}
       For GNEP \eqref{game_MI}, the cost functions $J_i(x_i,\mathbf{x}_{-i})$ for agent $i\in[I]$ are in the quadratic form as:
\begin{equation}\notag
    J_i(x_i,\mathbf{x}_{-i}) = \frac{1}{2}x_i^\top Q(\mathbf{x}_{-i})x_i + p^\top(\mathbf{x}_{-i})x_i + r(\mathbf{x}_{-i}),
\end{equation}
    and the strategy sets can be described by $H_i(\mathbf{x}_{-i})x_i\leq g_i(\mathbf{x}_{-i})$ in addition to \eqref{C2}.
    % $\quad(\overline{\bm{\beta}}\hat{\xi}+\overline{b}-\sum_{i\in[I]}\overline{\mathbf{A}}_{n_i}x_i)+\overline{M}\mathbf{q}\geq\tau\mathbf{e}_{mK}-\overline{E}\mathbf{s}$.
\end{assumption}
This assumption is reasonable as the quadratic cost function and polyhedron strategy set are compatible with multiple real-world problems (e.g., the energy trading problem between prosumers and the pricing problem between CSs).
For example, in the aggregative game where one agent's payoff depends on its own decision and an aggregation of all agents' decisions, the objective of each agent may also be described in a quadratic form.
% $i$ is in the form of:
% \begin{equation*}
%     J_i\left(x_i,\sum_{i=1}^Ix_i\right) = f_i\left(\sum_{i=1}^Ix_i\right)x_i+g_i(x_i).
% \end{equation*}
% where $f_i(\sum_{i}x_i)$ and $g_i(x_i)$ are the functions of $\sum_{i}x_i$ and $x_i$, respectively. In this regard, $J_i\left(x_i,\sum_{i=1}^Ix_i\right)$ can be expressed in the quadratic form if $f_i(\sum_{i}x_i)$ is a linear function of $\sum_{i}x_i$ and $g_i(x_i)$ is a quadratic function of $x_i$. 
\begin{theorem}
    Under Assumption \ref{ass:Q}, the GNEP \eqref{game_MI} is equivalent to the following MINLP problem:
\begin{equation}\label{MINLP}
    \begin{aligned}
        &\inf_{\mathbf{x}\in \mathcal{X}(\mathbf{x})} \hat{V}^c(\mathbf{x}) \\
          &= \inf_{\mathbf{x}\in \mathcal{X}(\mathbf{x}), \mathbf{\lambda}\geq 0}
          \left\{
          \sum_{i = 1}^{I}\left[\frac{1}{2}x_i^\top Q(\mathbf{x}_{-i})x_i \!+\! p^\top(\mathbf{x}_{-i})x_i \!+\! r(\mathbf{x}_{-i})\right]\right.\\
          &\left.\!+\!
          \sum_{i = 1}^{I}\!
          \left[\frac{1}{2}P_i^{*\top}(\lambda_i) Q_i^{-1}(\mathbf{x}_{-i})P_i^*(\lambda_i) \!+\! \lambda_i^\top g_i^*(\mathbf{x}_{-i})\! - \!r(\mathbf{x}_{-i})\right]\!\!\right\}\!.
    \end{aligned}
\end{equation}
where $\lambda_i = [\lambda_i^{a\top}, \lambda_i^{s\top}]^\top$, $g_i^*(\mathbf{x}_{-i}) = [g_i^\top(\mathbf{x}_{-i}), g_i^{s\top}]^\top$, 
$g_i^{s} = (\overline{\bm{\beta}}\hat{\xi}+\overline{b}-\sum_{j\in[I],j\neq i}\overline{\mathbf{A}}_{n_j}x_j)+\overline{M}\mathbf{q}-\tau\mathbf{e}_{mK}+\overline{E}\mathbf{s}$,
$P_i^*(\lambda_i) = p(\mathbf{x}_{-i})+ H_i^{*\top}(\mathbf{x}_{-i})\lambda_i$, $H_i^{*}(\mathbf{x}_{-i}) =[H_i^\top(\mathbf{x}_{-i}), \overline{\mathbf{A}}_{n_i}^\top]^\top$.
\end{theorem}
% Consider the cost functions $J_i(x_i,\mathbf{x}_{-i})$ for agent $i\in[I]$ are in the quadratic form as:
% \begin{equation}
%     J_i(x_i,\mathbf{x}_{-i}) = \frac{1}{2}x_i^\top Q(\mathbf{x}_{-i})x_i + p^\top(\mathbf{x}_{-i})x_i + r(\mathbf{x}_{-i}).
% \end{equation}
% It should be noted that for one kind of widely used game called the aggregative game, its payoffs depend on the agent's own strategy and an aggregate of strategy of all agents is usually in the form of:
% \begin{equation}
%     J\left(x_i,\sum_{i}x_i\right) = f\left(\sum_{i}x_i\right)x_i+g(x_i),
% \end{equation}
% which can be expressed in quadratic form if $f(\sum_{i}x_i)$ is a linear function of $\sum_{i}x_i$ and $g(x_i)$ is in the quadratic form. 

% After specifying the form of the cost function, we assume the strategy sets of agent $i\in[I]$ except for the shared constraint $\quad(\overline{\bm{\beta}}\hat{\xi}+\overline{b}-\sum_{i\in[I]}\overline{\mathbf{A}}_{n_i}x_i)+\overline{M}\mathbf{q}\geq\tau\mathbf{e}_{mK}-\overline{E}\mathbf{s}$, are in the form of $H_i(\mathbf{x}_{-i})x_i\leq g_i(\mathbf{x}_{-i})$.
\begin{proof}
   As \eqref{Vc} is naturally separable across all agents, we can write it as:
\begin{equation*}\label{Vc_ind}
    \begin{aligned}
        \hat{V}^c(\mathbf{x}) &= \sup_{\mathbf{y}\in \mathcal{X}^c(\mathbf{x})}\sum_{i = 1}^{I+1}\left[J^c_i(\mathbf{x})-J^c_i(y_i,\mathbf{x}_{-i})\right]\\
         &= \sum_{i = 1}^{I+1}J^c_i(\mathbf{x})-
         \inf_{\mathbf{y}\in \mathcal{X}^c(\mathbf{x})}\sum_{i = 1}^{I+1} J^c_i(y_i,\mathbf{x}_{-i})\\
         &= \sum_{i = 1}^{I+1}J^c_i(\mathbf{x})-
         \sum_{i = 1}^{I+1} \inf_{y_i\in \mathcal{X}_i^c(\mathbf{x}_{-i})}J^c_i(y_i,\mathbf{x}_{-i}).
         % \\
         % &= \sum_{i = 1}^{I+1}J^c_i(\mathbf{x})-
         % \sum_{i = 1}^{I+1} \sup_{\lambda_i\geq 0}J^d_i\\
         % &= \inf_{\mathbf{\lambda}\geq 0}\left[\sum_{i = 1}^{I}J^c_i(\mathbf{x})-
         % \sum_{i = 1}^{I} J^d_i\right],  
    \end{aligned}
\end{equation*}
Let $J^d_i$ be the cost function of the dual problem for the $i$th agent, which yields 
\begin{equation*}
    \begin{aligned}
        \hat{V}^c(\mathbf{x}) &=\sum_{i = 1}^{I+1}J^c_i(\mathbf{x})-
         \sum_{i = 1}^{I+1} \sup_{\lambda_i\geq 0}J^d_i\\
         &= \inf_{\mathbf{\lambda}\geq 0}\left[\sum_{i = 1}^{I}J^c_i(\mathbf{x})-
         \sum_{i = 1}^{I} J^d_i\right].
    \end{aligned}
\end{equation*}
As $J^c_{I+1}(\mathbf{x}) = 0$, we can simply change $\sup$ to $\inf$ without changing it into its dual problem. 
In view of \eqref{SolVc}, we obtain a one-layer optimization problem:
\begin{equation}\label{infV}
    \begin{aligned}
        &\inf_{\mathbf{x}\in \mathcal{X}(\mathbf{x})} \hat{V}^c(\mathbf{x})\\
        &=\inf_{\mathbf{x}\in \mathcal{X}(\mathbf{x})}\inf_{\bm{\lambda}\geq 0}\left[\sum_{i = 1}^{I}J^c_i(\mathbf{x})-
         \sum_{i = 1}^{I} J^d_i(\lambda_i)\right] \\
         &=\inf_{\mathbf{x}\in \mathcal{X}(\mathbf{x}), \bm{\lambda}\geq 0}\left[\sum_{i = 1}^{I}J^c_i(\mathbf{x})-
         \sum_{i = 1}^{I} J^d_i(\lambda_i)\right].
    \end{aligned}
\end{equation}
Under Assumption \ref{ass:Q}, the optimization problem of agent $i\in[I]$ becomes
\begin{equation*}
\left\{
    \begin{aligned}
    &\min_{x_i\in\mathbb{R}^{n_i}}\quad \frac{1}{2}x_i^\top Q(\mathbf{x}_{-i})x_i + p^\top(\mathbf{x}_{-i})x_i + r(\mathbf{x}_{-i})\\   
    &s.t.\quad H_i(\mathbf{x}_{-i})x_i\leq g_i(\mathbf{x}_{-i})\\  &\quad\quad(\overline{\bm{\beta}}\hat{\xi}+\overline{b}-\sum_{i\in[I]}\overline{\mathbf{A}}_{n_i}x_i)+\overline{M}\mathbf{q}\geq\tau\mathbf{e}_{mK}-\overline{E}\mathbf{s},
    \end{aligned}
    \right.
\end{equation*}
and its dual problem becomes
\begin{equation}\label{du_i}
\left\{
    \begin{aligned}
       \sup_{\mathbf{\lambda_i}\in \mathbb{R}^{n_i}}\!\! \!&\!-\!\frac{1}{2}P_i^{*\top}(\lambda_i) Q_i^{-1}(\mathbf{x}_{-i})P_i^*(\lambda_i)\! -\! \lambda_i^\top g_i^*(\mathbf{x}_{-i})\! + \!r(\mathbf{x}_{-i})\\
       &s.t.\quad \mathbf{\lambda_i}\geq 0.
    \end{aligned}
    \right.
\end{equation}
% where $\lambda_i = [\lambda_i^{a\top}, \lambda_i^{s\top}]^\top$, $H_i^{*}(\mathbf{x}_{-i}) =[H_i^\top(\mathbf{x}_{-i}), \overline{\mathbf{A}}_{n_i}^\top]^\top$, $g_i^*(\mathbf{x}_{-i}) = [g_i^\top(\mathbf{x}_{-i}), g_i^{s\top}]^\top$, 
% $g_i^{s} = (\overline{\bm{\beta}}\hat{\xi}+\overline{b}-\sum_{j\in[I],j\neq i}\overline{\mathbf{A}}_{n_j}x_j)+\overline{M}\mathbf{q}-\tau\mathbf{e}_{mK}+\overline{E}\mathbf{s}$,
% $P_i^*(\lambda_i) = p(\mathbf{x}_{-i})+ H_i^{*\top}(\mathbf{x}_{-i})\lambda_i$.
By substituting \eqref{du_i} into \eqref{infV}, we obtain the MINLP \eqref{MINLP}.
\end{proof}

Notice that the nonlinearity in the objective of \eqref{MINLP} arises from the cross terms of the continuous variables, and the integer variables do not couple with these continuous variables, thereby reducing the problem’s overall complexity.
\begin{remark}
    In the case that $Q(\mathbf{x}_{-i})$ is a constant matrix, $p(\mathbf{x}_{-i})$ and $r(\mathbf{x}_{-i})$ are linear functions of $\mathbf{x}_{-i}$. Equation \eqref{MINLP} can be reduced to a problem with a quadratic objective containing cross terms of continuous variables, which is a more tackleable form. When $Q(\mathbf{x}_{-i})$ and $p(\mathbf{x}_{-i})$ are both constant, $r(\mathbf{x}_{-i})$ is a linear function of $\mathbf{x}_{-i}$, \eqref{MINLP} can be reduced to a mixed-integer linear programming problem.
\end{remark}

%----------------------------------------
\section{Case Study}
The proposed method is examined in this section to solve a pricing problem among EV charging stations (CSs).
%\subsection{CS Pricing Problem}\label{eg_EV}
Consider a charging network with $I$ non-cooperative CSs. In each time slot, every CS independently sets its charging price to maximize its profit $J_i$, which is modeled as the revenue generated from providing charging services to EVs minus the cost incurred from purchasing electricity. As such, each CS determines the pricing strategy by solving  
%Due to the non-cooperative nature of this decision-making process, the interaction among CSs can be modeled as a generalized Nash game. By modeling $J_i$ as the revenue generated from providing charging services to electric vehicles (EVs) minus the cost incurred from purchasing electricity. Therefore, the objective of each CS can be defined as
% \begin{equation}
        %     \min_{c^s_i,u_i} J_i =  -c^s_iu_iE^{ev}+c^b(N^0_i+u_i)p^{EV},
% \end{equation}
\begin{equation*}
    \min_{c^s_i} -J_i =  -(c^s_i-c^b)(N^0_i+u_i)E_{d},
\end{equation*}
where $c^s_i$ and $c^b$ are the charging price set by CS and the purchasing price of the electricity, respectively. $u_i$ is arriving EVs, $N^0_i$ is the number of onsite EVs at the beginning of this time slot, and $E_{d}$ is the average electricity demand for each EV. 
The number of arriving EVs $u_i(t)$ is influenced by the charging price of the CS, as modeled by 
$
    u_i = u^0_{i} - \alpha^u(c^s_i-c^0),
$
with $u^0_{i}$ the nominal number of EV arrivals at CS $i$, $c^0$ being the nominal electricity purchasing price from the grid, and  $\alpha^u$ is a coefficient capturing the sensitivity of EV charging demand to price variations \cite{10.3389/fenrg.2023.1343311}. %The equation shows that as the charging price rises, fewer EVs charge at the station.
%The total charging demand for each time slot is within a certain range represented by $\underline{u}$ and $\overline{u}$. 
To ensure the overall charging demand is satisfied, we impose a coupled constraint on the total number of EV arrivals, which incentivizes charging stations to price strategically and maintain network efficiency
\begin{equation}\label{u_demand}
    \sum_{i=1}^I u_i \geq \underline{u}+ \delta_u,
\end{equation}
where $\underline{u}$ is the lower limit of overall charging demand and $\delta_u$ represents the uncertainty in demand arising from various factors, such as traffic conditions and weather. According to \cite{9353238}, the electricity price is linearly dependent on the demands of prosumers, as modeled by
% \begin{align}
%     c^b = c^0 + \alpha^c\sum_{i=1}^I(N_i^0+u_i)p^{EV},
% \end{align}
$
    c^b = c^0 + \alpha^c\sum_{i=1}^I(u_i-u_i^0),
$
where $\alpha^c$ is a coefficient representing the relation between the EV charging demand and electricity price. Moreover, the number of onsite EVs has to be restricted by
$
    0 \leq N_i^0+u_i \leq \overline{N}_i,
$
with $\overline{N}_i$ the capacity of CS $i$.
To deal with the uncertainties $\delta_u$, \eqref{u_demand} can be reformulated as a DRCC to ensure a low probability of constraint violation
\begin{equation*}\label{Pu}
    \inf_{\mathbb{P}\in \mathcal{P}}\mathbb{P}\left[ \alpha^u\mathbf{e}_I^\top\mathbf{c}^s\leq -\delta_u+\hat{u}\right]\geq 1-\epsilon,
\end{equation*}
where $\mathbf{c}^s = \begin{bmatrix}   c_1^s&c_2^s&\cdots&c_I^s
\end{bmatrix}^\top$, $\hat{u} =  -\underline{u}+\sum_{i=1}^I u_i^0+\alpha^uc^0I$, and $\mathcal{P}$ is the ambiguity set of $\delta_u$. 
%---The following is a draft for the coding---
% \subsection{Solution of CS problem}
% The objective function can be expressed as:
% \begin{equation}
%     J_i = \frac{1}{2}Q_ic_i^{s2} + p_i(\mathbf{c}^s_{-i})c_i^s + r_i(\mathbf{c}^s_{-i}),
% \end{equation}
% where $Q_i = (\alpha^p+1)\alpha^up^{EV}$, $r_i(\mathbf{c}^s_{-i}) = (-\alpha^p\sum_{j\neq i}c_j^s+c^0)(N_i^0+u_i^0)p^{EV}$, and $p_i(\mathbf{c}^s_{-i}) = [-\alpha^u(-\alpha^p\sum_{j\neq i}c_j^s+c^0)+(N_i^0+u_i^0)(-\alpha^p-1)]p^{EV}$, with $\alpha^p = \alpha^c\alpha^up^{EV}$.
% The DRCC shared constraint:
% $$\alpha^u\mathbf{e}_{I}^\top\mathbf{c}^s-du-\sum_{i=1}^I u_i^0+\overline{u}(t)+Mq_k\geq\tau-s_k$$
% $$-\alpha^u\mathbf{e}_{I}^\top\mathbf{c}^s+du+\sum_{i=1}^I u_i^0-\underline{u}(t)+Mq_k\geq\tau-s_k$$
%%---------------------------------------
%\subsection{Simulation Results}
%In this section, we apply the proposed method to solve the example introduced in Section \ref{eg_EV}. 
The resulting MINLP problem is run on AMD Ryzen Threadripper PRO 7965WX with 24 Cores, 48 Threads, 4.2GHz Base, 5.3GHz Turbo, and solved using Bonmin \cite{BONAMI2008186}. For numerical simulation, we consider a scenario of three identical charging stations with $u_0=50$, $\underline{u}=80$, $\alpha^u=500$, $\alpha^c=5e^{-4}$, $\overline{N} = 50$, $N_0=0$, and $c_0 = 0.12 \pounds$/kWh.

% \begin{table}[tbh!]
% \label{table:prmt}
% \centering
% \caption{Simulation Parameters}
% \vspace{-3mm}
% \begin{tabular}{|c|c||c|c|}
%  \hline
%  Number of agents & 3 & $c_0$ & 0.12 \pounds/kWh\\
%  \hline
%  $u_0$ & 50 & $\underline{u}$ & 80 \\
%  \hline
%  $\alpha^u$ & 500 & $\alpha^c$ & 5e-4 \\
%  \hline
%  $\overline{N}$ & 50 & $N_0$ & 0 \\
%  \hline
% \end{tabular}
% \end{table}
% \begin{figure}[htb!]
%     \centering
%     \includegraphics[width=1.0\linewidth]{Results/GNEP_N0.eps}\\
%     \caption{\label{fig:R1} The relation between the change of the GNE and the initial occupation for 1st charging station.}
% \end{figure}

% The impact of various parameters on the position of GNEs is assessed by varying $N_1^0$, $\underline{u}$, and $\epsilon$, respectively.
The results in Table \ref{table:GNE} indicate that $\underline{u}$ and $\epsilon$, associated with the shared constraint, do not usually affect the position of GNEs. However, if these parameters tighten the shared constraint excessively, then the GNEs will no longer exist.
\begin{table}[tbh!]
\centering
\caption{GNE position for different parameters}
\vspace{-3mm}
\begin{tabular}{|c|c|c|c|}
 \hline
 \multirow{2}*{$N_1^0$} & price for CS 1 & price for CS 2, 3 & \multirow{2}*{objective}\\
 ~ & (\pounds/kWh) & (\pounds/kWh) & ~\\
 \hline
 15 & 0.176944 & 0.160278 & 5.33E-15\\ 
 \hline
 30 & 0.192222 & 0.158889 & 8.88E-15\\ 
 \hline
 45 & 0.21 & 0.157073 & 0.0037 (Non-exist)\\ 
 \hline
 \hline
 $\underline{u}$ & \multicolumn{2}{c|}{price for each CS (\pounds/kWh)} & objective\\
 \hline
 50 & \multicolumn{2}{c|}{0.161667} & 2.66E-15\\
 \hline
 70 & \multicolumn{2}{c|}{0.161667} & 7.11E-15\\
 \hline
 90 & \multicolumn{2}{c|}{0.15999999} & 7.50E-03 (Non-exist)\\
 \hline
 \hline
 $\epsilon$ & \multicolumn{2}{c|}{price for each CS (\pounds/kWh)} & objective\\
 \hline
 0.1 & \multicolumn{2}{c|}{0.161667} & 4.44E-15\\
 \hline
 0.05 & \multicolumn{2}{c|}{0.161667} & 1.07E-14\\
 \hline
 0.01 & \multicolumn{2}{c|}{0.159799} & 0.00941894 (Non-exist)\\
 \hline
\end{tabular}
\label{table:GNE}
\end{table}
% We assess the impact of various parameters on the positioning of the GNE by varying the value of $N_1^0$ for the first CS, $\underline{u}$, and $\epsilon$, respectively. 
% As shown in Table \ref{table:GNE}, the values of $\underline{u}$ and $\epsilon$, both associated with the shared constraint, do not alter the position of the GNE. However, when these parameters lead to an excessively tight shared constraint, specifically, one that excludes the previously existing GNE, the GNE ceases to exist. 
On the other hand, increasing $N^0_1$ makes the first CS different from the other two CSs, leading to a change in the position of the GNE. As $N^0_1$ increases, the charging price at the first CS rises, while the prices at the other two decrease. This trend is intuitive, as a larger $N^0_1$ reduces available capacity, prompting the first CS to set a higher price to reduce the coming EVs. Furthermore, the existence of a GNE is tied to the objective function: a GNE exists when the objective attains zero.

Finally, to evaluate the scalability of the approach, 
%computational efficiency of solving the MINLP problem at different scales, 
we vary the number of CSs and compare the corresponding computation time and the GNE. The results are illustrated in Fig.~\ref{fig:fig1}, which shows that the computation time remains manageable even as the number of CSs increases to $500$. Notably, the GNE value decreases monotonically as the number of agents increases and converges eventually to the nominal electricity purchasing price, given a fixed number of vehicles.
\begin{figure}[htp!]
	\includegraphics[width=\columnwidth]{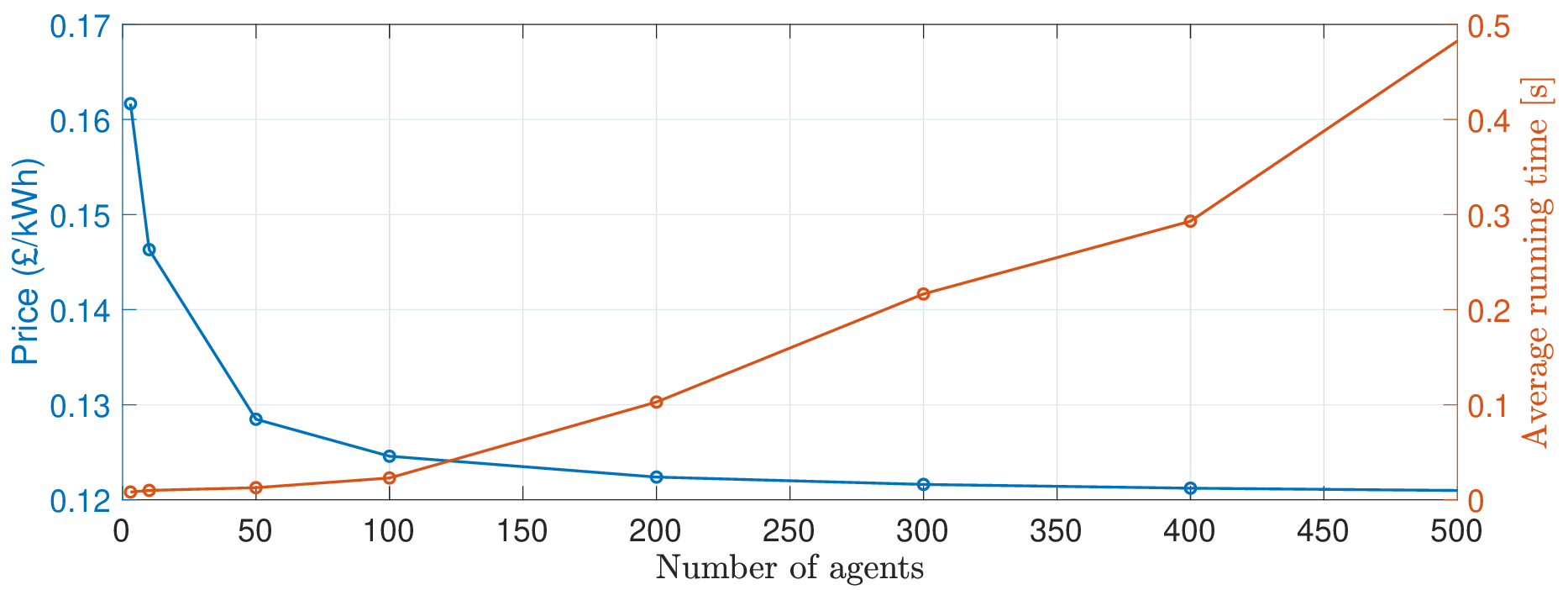}\\[-4.5ex]
	\centering
	\caption{Average computation time and the resulting GNE.}
	\label{fig:fig1}
\end{figure}
% In addition, the values of GNEs approach $c^0$ as the number of CS increases. This can be proved by assuming the GNE price of all CS is identical, and the optimal solution for each CS decreases as $I$ increases.
%\begin{table}[tbh!]
%\label{table:Sol_Time}
%\centering
%\caption{\footnotesize Value of GNE and Computation Time % for different numbers of agents
%}
%\vspace{-3mm}
%\begin{tabular}{|c|c|c|}
% \hline
% Number of agents & price (\pounds/kWh) & Computation Time (s) \\
% \hline
% 3 & 0.161667 & 0.008185387
% \\ 
% \hline
% 10 & 0.146316
% & 0.009926558 \\ 
% \hline
% 50 & 0.128475
% & 0.012771606 \\ 
% \hline
% 100 & 0.124587 & 0.022969246 \\ 
% \hline
%\end{tabular}
%\end{table}
%----------------------------------------------------------
\section{Conclusion}
This paper deals with the GNEP with joint DRCCs over a Wasserstein ball. An exact reformulation approach is proposed to rewrite the original computationally intractable problem into a deterministic form using the Nikaido-Isoda function. It is shown that the equilibrium of the GNEP can be obtained by solving a MINLP, provided that the individual agents’ objectives are quadratic in their respective variables. Since the integer and continuous variables are decoupled in the MINLP, the problem becomes significantly more computationally tractable compared to the general case, as verified by a case study on the charging station pricing problem. Furthermore, under certain convexity assumptions, we show that this MINLP can be solved efficiently, as demonstrated by a case study on the charging station pricing problem. One interesting future direction is to generalize the proposed method to  problems that cannot be transformed into their dual form.
% \textcolor{red}{add two sentences to discuss the future directions?}
\bibliographystyle{IEEEtran}
\bibliography{ref}
\end{document}